\documentclass[runningheads]{llncs}
\usepackage{graphicx}
\usepackage[linesnumbered,titlenotnumbered,ruled]{algorithm2e}
\usepackage{tikz,pgffor}
\usetikzlibrary{arrows,backgrounds}
\usepackage{tkz-euclide}
\usetikzlibrary{arrows}
\usetikzlibrary{shapes}
\usetikzlibrary{calc}
\usetikzlibrary{automata}
\usetikzlibrary{positioning}
\usepackage{listings}
\usepackage{transparent}
\usepackage{amsfonts}
\usepackage{amsmath}
\usepackage{calrsfs}
\DeclareMathAlphabet{\pazocal}{OMS}{zplm}{m}{n}
\usepackage{thm-restate}

\newcommand{\N}{\mathbb{N}}
\newcommand{\R}{\mathbb{R}}
\newcommand{\Q}{\mathbb{Q}}

\newcommand{\C}{\pazocal{C}}

\newcommand{\F}{\pazocal{F}}

\newcommand{\M}{\pazocal{M}}

\newcommand{\calL}{\pazocal{L}}

\newcommand{\sub}{\textit{sub}}
\newcommand{\fsub}{\textit{Fsub}}
\newcommand{\psub}{\textit{psub}}
\newcommand{\nsub}{\textit{nsub}}
\newcommand{\Deg}{\textit{deg}}
\newcommand{\cf}{\textit{cf}}
\newcommand{\nr}{\textit{Out}}
\newcommand{\oth}{\textit{Other}}
\newcommand{\DL}{\Delta(\mathcal{L})}

\newcommand{\EXPTIME}{\textbf{EXPTIME}}

\newcommand{\twoEXPSPACE}{\mbox{\rm\bf 2-EXPSPACE}}
\newcommand{\AP}{\textit{AP}}
\newcommand{\run}{\textit{Run}}
\newcommand{\NP}{\textbf{NP}}

\newcommand{\tran}[1]{\xrightarrow{#1}}

\tikzstyle{label}=[shape=circle,draw,inner sep=0pt,minimum size=5mm]
\tikzstyle{line} = [draw, color=black, -latex]
\tikzstyle{max}=[thick,draw,color=mybrown,minimum size=1.4em,inner sep=0em]
\tikzstyle{min}=[diamond,thick,draw,color=mybrown,minimum size=1.4em,%
inner sep=0em]
\tikzstyle{stoch}=[circle,thick,draw,minimum size=1.5em,%
inner sep=0em]
\tikzstyle{ran}=[circle,thick,draw,minimum size=1.4em,%
inner sep=0em]
\tikzstyle{mc}=[rounded corners,thick,draw,minimum size=1.4em,%
inner sep=.5ex]
\tikzstyle{tran}=[thick,draw,->,>=stealth,rounded corners]
\tikzstyle{loop left}=[tran, to path={.. controls +(150:.8) 
	and +(210:.8) .. (\tikztotarget) \tikztonodes}]
\tikzstyle{loop right}=[tran, to path={.. controls +(30:.8) 
	and +(330:.8) .. (\tikztotarget) \tikztonodes}]
\tikzstyle{loop above}=[tran, to path={.. controls +(60:.5) 
	and +(120:.5) .. (\tikztotarget) \tikztonodes}]
\tikzstyle{loop below}=[tran, to path={.. controls +(240:.8) 
	and +(300:.8) .. (\tikztotarget) \tikztonodes}]
\tikzstyle{bigstoch}=[circle,draw,minimum size=8ex,inner sep=0pt,font=\Large, very thick,text centered, fill=blue!20, fill opacity=0.2, draw=black!80, text opacity=1]
\tikzstyle{bigstoch2}=[circle,draw,minimum size=8ex,inner sep=0pt,font=\Large, very thick,text centered, fill=blue!20, fill opacity=0.2, draw=black!80, text opacity=1, double]
\tikzstyle{bigmin}=[diamond,draw,minimum size=8ex,inner sep=0pt,font=\Large, very thick,text centered, fill=blue!20, fill opacity=0.2, draw=black!80, text opacity=1]
\tikzstyle{bigtran}=[very thick,draw,-angle 60,font=\scriptsize, inner sep = 6pt]

\newcommand*{\pw}{\mathcal{P}}

\newcommand*{\nat}{\mathbb{N}}

\newcommand*{\m}{\mathbb{P}}

\newcommand*{\mc}{(S,P,v)}
\newcommand*{\fullmc}{M = (S,P,v)}

\newcommand*{\vp}{\varphi}

\newcommand*{\op}{\triangleright}

\newcommand*{\opx}{\operatorname{X}}
\newcommand*{\opu}{\operatorname{U}}
\newcommand*{\opf}{\operatorname{F}}
\newcommand*{\opg}{\operatorname{G}}

\newcommand*{\dl}{\Delta(\mathcal{L})}

\begin{document}
\title{The Satisfiability Problem for a Quantitave Fragment of PCTL}
%
%

\author{Miroslav Chodil \and
Anton\'{\i}n Ku\v{c}era\orcidID{0000-0002-6602-8028}}
\authorrunning{Chodil and Ku\v{c}era}
%
\institute{Masaryk University, Brno, Czechia\\ 
}
\maketitle              
\begin{abstract}
We give a sufficient condition under which every finite-satisfiable formula of a given PCTL fragment has a model with at most doubly exponential number of states (consequently, the finite satisfiability problem for the fragment is in $\twoEXPSPACE$). The condition is semantic and it is based on enforcing a form of ``progress'' in non-bottom SCCs contributing to the satisfaction of a given PCTL formula. We show that the condition is satisfied by PCTL fragments beyond the reach of existing methods.  

\keywords{Probabilistic temporal logics  \and Satisfiability \and PCTL.}
\end{abstract}



\section{Introduction}
\label{sec-intro}

Probabilistic CTL (PCTL) \cite{HJ:logic-time-probability-FAC} is a temporal logic applicable to discrete-time probabilistic systems with Markov chain semantics. PCTL is obtained from the ``standard'' CTL (see, e.g., \cite{Emerson:temp-logic-handbook}) by replacing the existential/universal path quantifiers with the probabilistic operator $P(\Phi) \bowtie r$. Here,  $\Phi$ is a path formula, $\bowtie$ is a comparison such as ${\geq}$ or ${<}$, and $r$ is a numerical constant. A~formula $P(\Phi) \bowtie r$ holds in a state $s$ if the probability of all runs initiated in $s$ satisfying~$\Phi$ is $\bowtie$-bounded by $r$.  The \emph{satisfiability problem for PCTL}, asking whether a given PCTL formula has a model, is a long-standing open question in probabilistic verification resisting numerous research attempts.

Unlike CTL and other non-probabilistic temporal logics, PCTL does not have a small model property guaranteeing the existence of a bounded-size model for every satisfiable formula. In fact, one can easily construct satisfiable PCTL formulae without \emph{any} finite model (see, e.g., \cite{BFKK:satisfiability}). Hence, the PCTL satisfiability problem is studied in two basic variants: (1) \emph{finite satisfiability}, where we ask about the existence of a finite model, and (2) \emph{general satisfiability}, where we ask about the existence of an unrestricted model. 

For the \emph{qualitative fragment} of PCTL, where the range of admissible probability constraints is restricted to $\{{=}0, {>}0, {=}1, {<}1\}$, both variants of the satisfiability problem are \EXPTIME-complete, and a finite description of a model for a satisfiable formula is effectively constructible \cite{BFKK:satisfiability}. Unfortunately, the underlying proof techniques are not applicable to general PCTL with unrestricted (quantitative) probability constraints such as ${\geq} 0.25$ or ${<}0.7$. 

To solve the \emph{finite} satisfiability problem for some PCTL fragment, it suffices to establish a computable upper bound on the size (the number of states) of a model for a finite-satisfiable formula of the fragment\footnote{Although there are uncountably many Markov chains with $n$ states, the edge probabilities can be represented symbolically by variables, and the satisfiability of a given PCTL formula in a Markov chain with $n$ states can then be encoded in first-order theory of the reals. For the sake of completeness, we present this construction in Appendix~\ref{app-encoding}.}. At first glance, one is tempted to conjecture the existence of such a bound for the whole PCTL because there is no apparent way how a finite-satisfiable PCTL formula $\varphi$ can ``enforce'' the existence of $F(|\varphi|)$ pairwise different states in a model of $\varphi$, where $F$ grows faster than any computable function. Interestingly, this conjecture is \emph{provably wrong} in a slightly modified setting where we ask about finite PCTL satisfiability in a \emph{subclass} of Markov chains $\M^k$ where every state has at most $k \geq 2$ immediate successors (the~$k$~is an arbitrarily large fixed constant). This problem is \emph{undecidable} and hence no computable upper bound on the size of a model in $\M^k$ exists  \cite{BFKK:satisfiability} (see \cite{BFKK:satisfiability-report} for a full proof). So far, all attempts at extending the undecidability proof of \cite{BFKK:satisfiability} to the class of unrestricted Markov chains have failed; it is not yet clear whether the obstacles are invincible. 

Regardless of the ultimate decidability status of the (finite) PCTL satisfiability, the study of PCTL fragments brings  important insights into the structure end expressiveness of PCTL formulae. The existing works \cite{KR:PCTL-unbounded,CHK:PCTL-simple} identify several fragments where every (finite) satisfiable formula has a model of bounded size and specific shape. In \cite{CHK:PCTL-simple}, it is shown that every formula $\varphi$ of the \emph{bounded fragment} of PCTL, where the validity of $\varphi$ in a state $s$ depends only on a bounded prefix of a run initiated in $s$, has a bounded-size tree model. In \cite{KR:PCTL-unbounded}, seven syntactically incomparable PCTL fragments based on $F$ and $G$ operators are studied. For each of these fragments, it is shown that every satisfiable (or finite-satisfiable) formula has a bounded-size model such that every non-bottom SCC is a singleton. It is also shown that there are finite-satisfiable PCTL formulae without a model of this shape. An example of such a formula is  
\[
   \psi \quad \equiv \quad G_{=1} \big(F_{\geq 0.5}(a \wedge F_{\geq 0.2} \neg a) \vee a\big) \ \wedge \ F_{=1}\, G_{=1}\,a \ \wedge\  \neg a
\]
In \cite{KR:PCTL-unbounded}, it is shown that $\psi$ is finite satisfiable\footnote{In \cite{KR:PCTL-unbounded}, the formula $\psi$ has the same structure but uses qualitative probability constraints.}, but every finite model of $\psi$ has a non-bottom SCC with at least two states, such as the Markov chain $M$ of Fig.~\ref{fig-model-psi}. 

\smallskip
\noindent
\textbf{Our contribution.} A crucial step towards solving the finite PCTL satisfiability is understanding the role of non-bottom SCCs. Intuitively, if a given PCTL formula $\varphi$ enforces a model with a non-bottom SCC, then the top SCC  must achieve some sort of ``progress'' in satisfying $\varphi$, and successor SCCs are required to satisfy only some ``simpler'' formulae. In this paper, we develop this intuition into an algorithm deciding finite satisfiability for various PCTL fragments beyond the reach of existing methods.   

More concretely, we design a sufficient condition under which every formula in a given PCTL fragment has a bounded model with at most doubly exponential number of states (consequently, the finite satisfiability problem for the fragment is in $\twoEXPSPACE$). The condition says that a progress in satisfying $\varphi$ is achievable by a SCC $C$ where $C$ has bounded number states and takes the form of a loop with one exit state of bounded outdegree (see Fig.~\ref{fig-loop}). Furthermore, the successor states are required to satisfy PCTL formulae strictly simpler than $\varphi$ in a precisely defined sense. Hence, a bounded model for these formulae exists by induction hypothesis, and thus we complete the construction of a bounded model for $\varphi$.

The above sufficient condition is ``semantic'' and it is satisfied by various mutually incomparable syntactic fragments of PCTL that are not covered by the methods of \cite{KR:PCTL-unbounded} (two of these fragments contain the formula $\psi$ presented above). Hence, our semantic condition can be seen as a ``unifying principle'' behind these concrete decidability results. 

In our construction, we had to address fundamental issues specific to quantitative PCTL. The basic observation behind the small model property proofs for non-probabilistic temporal logics (and also \emph{qualitative} PCTL) is that the satisfaction of a given formula in a given state $s$ is determined by the satisfaction of $\varphi$ and its subformulae in the successor states of~$s$. For (quantitative) PCTL, this is not true. For example, if we know whether or not the immediate successors of a state $s$ satisfy the formula $\opf_{\geq 0.2} \varphi$, we cannot say anything about the (in)validity of this formula in~$s$. What we need is a \emph{precise probability} of satisfying the path formula $\opf \varphi$ in the successors of~$s$. Clearly, it has no sense to filter a model according to the satisfaction of infinitely many formulae of the form $\opf_{\geq r} \varphi$. In our proof, we invent a method for extending the set of ``relevant formulae'' so that it remains bounded and still captures the crucial properties of states.

The methodology presented in this paper can be extended by considering SCCs with increasingly complex structure and analyzing the achievable ``progress in satisfaction'' of PCTL formulae (see Section~\ref{sec-concl} for more comments). We believe that this effort may eventually result in a decidability proof for the whole PCTL.

\smallskip
\noindent
\textbf{Related work.} The satisfiability problem for non-probabilistic CTL is known to be \EXPTIME-complete in \cite{EH:CTL-satisfiability}. The same upper bound is valid also for a richer logic of the modal $\mu$-calculus \cite{BB:temp-logic-fixed-points,FL:PDL-regular-programs}. The probabilistic extension of CTL (and also CTL$^*$) was initially studied in its qualitative form \cite{LS:time-chance-IC,HS:Prob-temp-logic}. The satisfiability problem is shown decidable in these works. A precise complexity classification of general and finite satisfiability, together with a construction of a (finite description of) a model is given in \cite{BFKK:satisfiability}. In the same paper, it is also shown that the satisfiability problem is undecidable when the class of admissible models is restricted to Markov chain with a $k$-bounded branching degree, where $k \geq 2$ is an arbitrary constant. A variant of the bounded satisfiability problem, where transition probabilities are restricted to $\{\frac{1}{2},1\}$, is proven \NP-complete
in \cite{BFS:bounded-PCTL}. The decidability of finite satisfiability for fragments of quantitative PCTL is established in the works \cite{KR:PCTL-unbounded,CHK:PCTL-simple} discussed above. 

The \emph{model-checking} problem for PCTL has been studied both for finite Markov chains (see, e.g., \cite{BK:PCTL-fairness,BA:MDP-PCTL,HK:quantitative-mu-calculus-LICS,BK:book}) and for infinite Markov chains generated by probabilistic pushdown automata and their subclasses \cite{EKM:prob-PDA-PCTL-LMCS,BKS:pPDA-temporal,EY:RMC-LTL-complexity-TCL}. The undecidability results for (finite) PCTL satisfiability in subclasses of Markov chains with bounded branching degree follow from the undecidability results for MDPs with PCTL objectives \cite{BBFK:Games-PCTL-objectives}.

\begin{figure}[t]
\centering
\begin{tikzpicture}[scale=1, every node/.style={scale=1}, x=1.3cm, y=1.3cm, font=\small]
   \node [stoch] (t) at (0,0)  {$a$};
   \node [stoch] (s) at (1,0)  {$\neg a$};
   \node [stoch] (u) at (0,-1) {$a$};
   \node [stoch,draw=none] (l) at (1.3,0) {$s$};
   \draw [tran]  (t.45) to node[above] {$0.6$} (s.135); 
   \draw [tran]  (s.225) to node[below] {$1$}  (t.315); 
   \draw [tran]  (t) to node[left] {$0.4$}     (u); 
   \draw [tran,loop right] (u) to node[right] {$1$} (u); 
\end{tikzpicture}
\caption{A Markov chain $M$ such that $s \models \psi$.}
\label{fig-model-psi}
\end{figure}
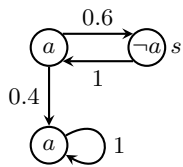

\section{Preliminaries}
\label{sec-prelim}

We use $\N$, $\Q$, $\R$ to denote the sets of non-negative integers, rational numbers, and real numbers,
respectively. We use the standard notation for writing intervals of real numbers, e.g., $[0,1)$ denotes the set of all  $r\in \mathbb{R}$ where $0 \leq r < 1$.




The logic PCTL \cite{HJ:logic-time-probability-FAC} is a probabilistic version of Computational Tree Logic  \cite{Emerson:temp-logic-handbook} obtained by replacing the existential and universal path quantifiers with the probabilistic operator $P(\Phi) \bowtie r$, where $\Phi$ is a path formula, $\bowtie$ is a comparison, and $r \in [0,1]$ is a constant. 

In full PCTL, the syntax of path formulae is based on the $\opx$ and $\opu$ (`next' and `until') operators. In this paper, we consider a variant of PCTL based on $\opf$ and $\opg$ operators. These operators are simplified forms of $\opu$, and capture the core of PCTL expressive power (see \cite{BK:book}).


\begin{definition}[PCTL]
\label{def-pctl}
   Let $\AP$ be a set of atomic propositions. The syntax of PCTL state and path formulae is defined by the following abstract syntax equations:
   \[
   \begin{array}{lcl}
      \varphi & ~~::=~~ & a \mid \neg a \mid \vp \wedge \vp \mid 
              \vp \vee \vp \mid P(\Phi) \op r\\
      \Phi & ::= &\opf \vp \mid \opg \vp 
   \end{array} 
   \]    
   Here, $a \in \AP$, $\op \in \{{\geq}, {>}\}$, and $r \in [0,1]$. The trivial probability constraints are 
\end{definition}
For the sake of simplicity, the trivial probability constraints `${\geq} 0$' and `${>}1$' are syntactically forbidden.  Since the formula $\Phi$ in the probabilistic operator $P(\Phi) \op r$ is always of the form $\opf\vp$ or  $\opg \vp$, we often write just $\opf_{\op r}\, \vp$ and $\opg_{\op r}\, \vp$
instead of $P(\opf\vp) \op r$ and $P(\opg\vp) \op r$, respectively. The probability constraint `${\geq} 1$' is usually written as `${=}1$'. The set of all state sub-formulae of a given state formula $\varphi$ is denoted by $\sub(\varphi)$. For a set $X$ of PCTL formulae, we use $\sub(X)$ to denote  $\bigcup_{\varphi \in X} \sub(\varphi)$.

Observe that the negation is applicable only to atomic propositions and the comparison ranges only over $\{{\geq}, {>}\}$. This is not restrictive because negations can be pushed inside, and formulae such as $\opf_{\leq r}\, \vp$ and $\opg_{< r}\, \vp$ are equivalent to $\opg_{> 1{-}r}\, \neg \vp$ and $\opf_{\geq 1{-}r}\, \neg \vp$, respectively.

PCTL formulae are interpreted over Markov chains where every state $s$ is assigned a subset $v(s) \subseteq \AP$ of atomic propositions valid in~$s$.
  
\begin{definition}[Markov chain]
    A Markov chain is a triple $M = (S,P,v)$, where $S$ is a finite or countably infinite set of \emph{states}, $P \colon S \times S \rightarrow [0,1]$ is a function such that $\sum_{t \in S} P(s,t)=1$ for every $s \in S$, and $v
    \colon S \rightarrow 2^{\AP}$.
\end{definition}

A \emph{path} in $M$ is a finite sequence $w = s_0 \ldots s_n$ of states such that \mbox{$P(s_i,s_{i+1}) > 0$} for all $i <n$. A \emph{run} in $M$ is an infinite sequence $\pi = s_0 s_1 \ldots$ of states such that every finite prefix of $\pi$ is a path in $M$. We also use $\pi(i)$ to denote the state $s_i$ of~$\pi$.

A \emph{strongly connected component (SCC)} of $M$ is a maximal $U \subseteq S$ such that, for all $s,t \in U$, there is a path from $s$ to $t$. A \emph{bottom SCC (BSCC)} of $M$ is a SCC $U$ such that for every $s \in U$ and every path $s_0\ldots s_n$ where $s = s_0$ we have that $s_n \in U$.

For every path $w = s_0 \ldots s_n$, let $\run(w)$ be the set of all runs starting with~$w$, and let $\m(\run(w)) = \prod_{i=0}^{n-1} P(s_i,s_{i+1})$. To every state $s$, we associate the probability space $(\run(s),\F_s,\m_s)$, where $\F_s$ is the $\sigma$-field generated by all $\run(w)$ where $w$ starts in $s$, and $\m_s$ is the unique probability measure obtained by extending $\m$ in the standard way (see, e.g., \cite{Billingsley:book}).

The \emph{validity} of a PCTL state/path formula for a given state/run of $M$ is defined inductively as follows:
\[
\begin{array}{lcl}
  s \models a & \mbox{ ~~iff~~ } & a \in v(s),\\
  s \models \neg a &  \mbox{ iff } & a \not \in v(s),\\
  s \models \vp_1 \wedge \vp_2 &  \mbox{ iff } & s \models \vp_1 \mbox{ and } s \models \vp_2,\\
  s \models \vp_1 \vee \vp_2 &  \mbox{ iff } & s \models \vp_1 \mbox{ or } s \models \vp_2,\\
  s \models P(\Phi) \op r  &  \mbox{ iff } &  \m_s(\{ \pi \in Run(s) \mid \pi \models \Phi \}) \op r,\\[1ex]
  \pi \models \opf \vp  &  \mbox{ iff } & \pi(i) \models \vp \mbox{ for some } i \in \N,\\
  \pi \models \opg \vp  &  \mbox{ iff } & \pi(i) \models \vp \mbox{ for all } i \in \N.
\end{array}
\]
For a set $X$ of PCTL state formulae, we write $s \models X$ iff $x \models \varphi$ for every $\varphi \in X$.

We say that $M$ is a \emph{model} of $\varphi$ if $s \models \varphi$ for some state $s$ of $M$. The \emph{(finite) PCTL satisfiability problem} is the question whether a given PCTL formula has a (finite) model.

A \emph{PCTL fragment} is a set of PCTL state formulae $\calL$ closed under state subformulae and changes in probability constraints, i.e., if $\opf_{\op r} \varphi \in \calL$ (or \mbox{$\opg_{\op r} \varphi \in \calL$}), then $\opf_{\geq r'} \varphi \in \calL$ (or $\opg_{\geq r'} \varphi \in \calL$) for every  $r'$.

\section{Results}
\label{sec-results}

In this section, we formulate our main results. 
As a running example, we use the formula $\psi$ of Section~\ref{sec-intro} and its model of Fig.~\ref{fig-model-psi}.

\begin{definition}
\label{def-closure}
  Let $\psi$ be a PCTL formula and $s$ a state in a Markov chain such that $s \models \psi$. The \emph{closure of $\psi$ in $s$}, denoted by $C_s(\psi)$, is the least set $K$ satisfying the following conditions:
  \begin{itemize}
     \item $\psi \in K$;
     \item if $\varphi_1 \vee \varphi_2 \in K$ and $s \models \varphi_1$, then $\varphi_1 \in K$;
     \item if $\varphi_1 \vee \varphi_2 \in K$ and $s \models \varphi_2$, then $\varphi_2 \in K$;
     \item if $\varphi_1 \wedge \varphi_2 \in K$, then $\varphi_1,\varphi_2 \in K$;
     \item if $\opf_{\op r} \varphi \in K$ and $s \models \varphi$, then $\varphi \in K$;
  \end{itemize} 
  Furthermore, for a finite set $X$ of PCTL formulae such that $s \models X$, we put $C_s(X) = \bigcup_{\psi \in X} C_s(\psi)$.
\end{definition}

Observe that $C_s(\psi)$ contains some but not necessarily \emph{all} subformulae of $\psi$ that are valid in $s$. In particular, there is no rule saying that if $\opg_{\op r} \varphi \in K$, then $\varphi \in K$. As we shall see, the subformulae within the scope of $\opg_{\op r}$ operator need special treatment. 


\begin{example}
For the formula $\psi$ and the state $s$ of our running example, we obtain
\[
   C_s(\psi) \quad = \quad \{\psi,\quad \opg_{=1} \big(\opf_{\geq 0.5}(a \wedge \opf_{\geq 0.2} \neg a) \vee a\big),\quad \opf_{=1}\, \opg_{=1}\,a ,\quad  \neg a\}
\] 
Observe that although $\opf_{=1}\, \opg_{=1}\,a \in \C_s(\psi)$, the formula $\opg_{=1}\,a$ is not included into  $\C_s(\psi)$ because $s \not\models \opg_{=1}\,a$.
\end{example}

The set $C_s(\psi)$ does not give a precise information about the satisfaction of relevant path formulae. Therefore, we allow for ``updating'' the closure with precise quantities.

\begin{definition}
\label{def-update}
   Let $X$ be a set of PCTL formulae and $s$ a state in a Markov chain such that $s \models \varphi$ for every $\varphi \in X$. The \emph{update of $X$ in $s$}, denoted by $U_s(X)$, is the set of formulae obtained by replacing every formula of the form $P(\Phi) \op r$ in $X$ with the formula $P(\Phi) \geq r'$, where $r' = \m_s(\{ \pi \in Run(s) \mid \pi \models \Phi \})$.
\end{definition}

Observe that $r' \geq r$, and the formulae of $X$ which are \emph{not} of the form  $P(\Phi) \op r$ are left unchanged by $U_s$. The $UC_s$ operator is defined by $UC_s(X) = U_s(C_s(X))$. Observe that $UC_s$ is idempotent, i.e., $UC_s(UC_s(X)) = UC_s(X)$.

\begin{example}
In our running example, we have that
\[
   UC_s(\psi) \quad = \quad \{\psi,\quad \opg_{=1} \big(\opf_{\geq 0.5}(a \wedge \opf_{\geq 0.2} \neg a) \vee a\big),\quad \opf_{= 1}\, \opg_{=1}\,a ,\quad  \neg a\}
\]   
because the probability constraint $r$ in the two formulae of the form $P(\Phi) \op r$ is equal to $1$ and cannot be enlarged.
\end{example}


In our next definition, we introduce a sufficient condition under which the finite satisfiability problem is decidable in a PCTL fragment $\calL$. 
\begin{definition}
\label{def-loop}
We say that a PCTL fragment $\calL$ is \emph{progressive} if for every finite set $X$ of PCTL formulae and every state $s$ of a finite Markov chain such that
\begin{itemize}
   \item $s \models X$,
   \item $X$ is closed and updated (i.e., $X = UC_s(X)$),
   \item $X \subseteq \calL$
\end{itemize}
there exists a \emph{progress loop}, i.e., a finite sequence $\mathcal{L} = L_0,\ldots,L_n$ of subsets of $\sub(X)$ satisfying the following conditions: 
\begin{itemize}
  \item[(1)] $X \subseteq L_i$ for some $i \in \{0,\ldots,n\}$;
  \item[(2)] $L_0,\ldots,L_n$ are pairwise different (this induces an upper bound on $n$);
  \item[(3)] for every $i \in \{0,\ldots,n\}$, we have that
  \begin{itemize} 
  \item if $a \in L_i$, then $\neg a \not\in L_i$;
  \item if $\varphi_1 \wedge \varphi_2 \in L_i$, then $\varphi_1,\varphi_2 \in L_i$;
  \item if $\varphi_1 \vee \varphi_2 \in L_i$, then $\varphi_1 \in L_i$ or $\varphi_2 \in L_i$;
  \item if $\opg_{\op r} \varphi \in L_i$, then $\varphi \in L_j$ for every $j \in \{0,\ldots,n\}$.
  \end{itemize}
\end{itemize}
Furthermore, let $\DL$ be the set of all $\varphi \in L_0 \cup \cdots \cup L_n$ such that one of the following conditions holds:
  \begin{itemize}
      \item $\varphi \equiv \opg_{\op r} \psi$;
      \item $\varphi \equiv \opf_{\op r} \psi$, $\psi \not\in L_0 \cup \cdots \cup L_n$;
      \item $\varphi \equiv \opf_{=1} \psi$, $\opf_{=1} \psi \in L_i$ for some $i$ such that $\psi \not\in L_i \cup \cdots \cup L_n$.
  \end{itemize}
We require that
\begin{itemize}
  \item[(4)] $s \models \DL$;
  \item[(5)] $s \not \models \psi$ for every formula of form $\opf_{\op r} \psi$ such that $\opf_{\op r} \psi \in \DL$;
  \item[(6)] $\cf_s(\DL) \subseteq \cf_s(X)$.\\
  Here, the set $\cf_s(Y)$ consists of all formulae $\opf \varphi$ such that $Y$ contains a formula of the form $\opf_{\op r} \varphi$, $s \not\models \varphi$, and there is a finite path from $s$ to a state~$t$ where $t \models \varphi$ and $t \not\models \opg_{=1}\psi$ for every formula $\opg \psi$ such that $\sub(Y)$ contains a formula of the form  $\opg_{\op r} \psi$ and $s \not\models \opg_{=1} \psi$. 
\end{itemize} 
\end{definition}

\begin{example}
\label{exa-loop}
  In our running example, consider $X = UC_s(\psi)$ . Then $L_0,L_1,L_2$, where
  \[
  \begin{array}{lcl}
     L_0 & = & \{ \psi,\ \opg_{=1} \big(\opf_{\geq 0.5}(a \wedge \opf_{\geq 0.2} \neg a) \vee a\big),\ 
                \opf_{\geq 0.5}(a \wedge \opf_{\geq 0.2} \neg a) \vee a,\\
         & & \  \opf_{\geq 0.5}(a \wedge \opf_{\geq 0.2} \neg a),\ \opf_{= 1} \opg_{=1}a ,\ \neg a \} \\
     L_1 & = & \{ \opf_{\geq 0.5}(a \wedge \opf_{\geq 0.2} \neg a) \vee a,\ a\}\\
     L_2 & = & \{ \opf_{\geq 0.5}(a \wedge \opf_{\geq 0.2} \neg a) \vee a,\ \opf_{\geq 0.5}(a \wedge \opf_{\geq 0.2} \neg a), a \wedge \opf_{\geq 0.2}\neg a,\ a,\ \opf_{\geq 0.2}\neg a \}\\
  \end{array}
  \]  
  is a progress loop for $X$ and $s$. Observe that 
  \[
   \DL = \{\opg_{=1} \big(\opf_{\geq 0.5}(a \wedge \opf_{\geq 0.2} \neg a) \vee a\big),\
                        \opf_{= 1} \opg_{=1}a \}\,.
  \]
  Furthermore, $X \subseteq L_0$ and $\DL \subseteq X$. 
\end{example}

\begin{figure}[t]
\centering
\begin{tikzpicture}[scale=1, every node/.style={scale=1}, x=1.5cm, y=1.5cm, font=\small]
   \node [stoch] (l1) at (0,0)  {$\ell_0$};
   \node [stoch] (l2) at (1,0)  {$\ell_1$};
   \node [stoch,draw=none] (l3) at (2,0) {};
   \node [stoch,draw=none] (l4) at (2.5,0) {};
   \node [stoch] (ln) at (3.5,0) {$\ell_n$};
   \node [stoch] (t1) at (4.5,.5)   {$t_1$};
   \node [stoch] (tm) at (4.5,-.5)  {$t_m$};
   \draw [tran]  (l1) to node[below] {$1$} (l2); 
   \draw [tran]  (l2) to node[below] {$1$} (l3); 
   \draw [tran]  (l4) to node[below] {$1$} (ln); 
   \draw [tran,-,dotted]  (l3) to (l4); 
   \draw [tran]  (ln) to node[above] {$x_1$} (t1); 
   \draw [tran]  (ln) to node[below] {$x_m$} (tm);
   \draw [tran]  (ln) -- +(0,.6) -- node[above] {$1 - \sum_{x_i}$} +(-3.5,0.6) -- (l1);
   \draw [tran,-,dotted]  (t1)+(0,-.3) to +(0,-.7);
   \node [draw=none,right=.2 of t1] (X1) {$X_{t_1} = UC_{t_i}(\DL_{t_i})$};
   \node [draw=none,right=.2 of tm] (Xm) {$X_{t_m} = UC_{t_m}(\DL_{t_m})$};
\end{tikzpicture}
\caption{A graph for a progress loop $L_0,\ldots,L_n$.}
\label{fig-loop}
\end{figure}
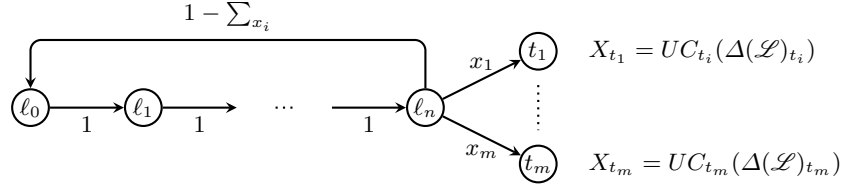

Intuitively, a progress loop allows to prove the existence of a bounded-size model for a finite-satisfiable formula $\psi$ where $\psi \in \calL$. Let us fix some (unspecified) finite Markov chain $M$ and a state $s$ of $M$ such that $s \models \psi$.  Initially, we put $X = UC_s(\psi)$. Then, we consider a progress loop $\mathcal{L} = L_0,\ldots,L_n$ for $X$ and~$s$, and construct the graph of Fig.~\ref{fig-loop}. The states $\ell_0,\ldots,\ell_n$ correspond to $L_0,\ldots,L_n$, and, as we shall see, $\ell_i \models L_i$ for every $i \in \{0,\ldots,n\}$ after completing our construction. Intuitively, the set $\DL$ contains formulae that are not satisfied by the loop itself, and must be ``propagated'' to the successors of $\ell_n$. The probabilities $x_1,\ldots,x_m$ are chosen so that $1 - \sum x_i$ is larger than the maximal $r \neq 1$ appearing in formulae of the form $\opf_{\op r} \varphi \in L_0 \cup \cdots \cup L_n$. This ensures that every $\ell_i$ visits every $\ell_j$ with high-enough probability. 

The loop is required not to spoil relevant formulae of the form $\opg_{\op r} \psi$ (see the last condition in (3)) and to satisfy almost all of the ``new'' formulae of the form  $\opf_{\op r} \psi$ that do not appear in $X$ and have been added to the loop because of the last condition in~(3). In Example~\ref{exa-loop}, such a ``new'' formula is, e.g.,  \mbox{$\opf_{\geq 0.5}(a \wedge \opf_{\geq 0.2} \neg a)$}. The ``new'' formulae that are not satisfied by the loop must satisfy the technical condition~(6) whose purpose is to ensure progress with respect to the measure defined in Section~\ref{sec-measure}. In Example~\ref{exa-loop}, there is no such formula, because \mbox{$\DL \subseteq X$}.

Now we explain how the successors $t_1,\ldots,t_m$ of $\ell_n$ are constructed, and what is the bound on~$m$. Recall that 
\[
   \DL = \big\{P(\Phi_1) \op r_1,\ \ldots,\ P(\Phi_v) \op r_v,\ \ldots,\ P(\Phi_u) \op r_u \big\} 
\]
where $\Phi_i \equiv \opf \varphi_i$ for $1\leq i \leq v$, and $\Phi_i \equiv \opg \varphi_i$ for $v < i \leq u$, respectively. Clearly, $u$ is bounded by the number of subformulae of the considered formula~$\psi$. For every state $t$ of $M$, let $\alpha_t$ be the $v$-dimensional vector such that 
\[
\alpha_t(i) = \m_s\big(\{ \pi \in Run(t) \mid \pi \models \Phi_i \}\big).
\]
Furthermore, let $B$ be the set of all states $t$ of $M$ such that $t$ either belongs to a BSCC of $M$ or $t \models \varphi_i$ for some $1 \leq i \leq v$. Since $M$ is finite, $B$ is also finite, but there is no upper bound on the size of $B$. For every $t \in B$, let $y_t$ be the probability of all runs initiated in $s$ visiting the state $t$ so that all states preceding the first visit to $t$ are not contained in $B$. We show that 
\[
   \alpha_s \quad \leq \quad \sum_{t \in B} y_t \cdot \alpha_t 
\]
Since $\sum_{t \in B} y_t = 1$, we can apply Carath\'{e}odory's theorem and thus obtain a subset $B' \subseteq B$ with at most $u+1$ elements such that $\sum_{t \in B} y_t \cdot \alpha_t$ lies in the convex hull of $\alpha_t$, $t \in B'$. That is,
\[
   \alpha_s \quad \leq \quad  p_1 \cdot \alpha_{t_1} + \cdots + p_m \cdot \alpha_{t_m}
\]
where $m \leq u+1$, $0 < p_i \leq 1$ for all $i \in \{1,\ldots,m\}$, $\sum_{i=1}^m p_i = 1$, and $B' = \{t_1,\ldots,t_m\}$. Let $\varrho > 0$ be a constant such that $1-\varrho \geq r$ for every $r \neq 1$  appearing in formulae of the form $\opf_{\op r} \varphi \in L_0 \cup \cdots \cup L_n$. For every $i \in \{1,\ldots,m\}$, the probability $x_i$ (see Fig.~\ref{fig-loop}) is defined by $x_i = p_i \cdot \varrho$.
Furthermore, for every $t_i \in B'$, we construct the set 
\[
  \DL_{t_i} = \big\{P(\Phi_1) \op \alpha_{t_i}(1),\ \ldots,\ P(\Phi_v) \op \alpha_{t_i}(v),\ \ldots,\ P(\Phi_u) \op \alpha_{t_i}(u) \big\} 
\]
and then the set $X_{t_i} = UC_{t_i}(\DL_{t_i})$. We have that $X_{t_i} \subseteq \calL$, and $X_{t_i}$ is \emph{smaller} than $X$ with respect to the measure defined in Section~\ref{sec-measure}. Hence, we proceed by induction, and construct a finite model of bounded size for $X_{t_i}$ by considering a progress loop for $X_{t_i}$ and $t_i$. Thus, we obtain the following theorem:

\begin{theorem}
\label{thm-main}
   Let $\calL$ be a progressive PCTL fragment. Then every finite-satisfiable formula $\psi \in \calL$ has a model with at most $a^{a^{a+5}}$ states where $a = |\sub(\psi)|$, such that every non-bottom SCC is a simple loop with one exit state (see Fig.~\ref{fig-loop}). Consequently, the finite satisfiability problem for $\calL$ is in \twoEXPSPACE. 
\end{theorem}

\noindent
A full technical proof of Theorem~\ref{thm-main} formalizing the above sketch is given in Appendix~\ref{app-main-proof}. The $\twoEXPSPACE$ upper bound is obtained by encoding the bounded satisfiability into existential theory of the reals. For the sake of completeness, we recall this encoding in Appendix~\ref{app-encoding}.

Theorem~\ref{thm-main} can be applied to various PCTL fragments by demonstrating their progressivity, and can be interpreted as a ``unifying principle'' behind these concrete decidability results. To illustrate this, we give examples of progressive fragments in Section~\ref{sec-fragments}.

\subsection{Progress measure}
\label{sec-measure}

A crucial ingredient of our result is a function measuring the complexity of PCTL formulae. The value of this function, denoted by $\| \cdot \|_s$, is strictly decreased by every progress loop, i.e., 
$\|X_{t_i}\|_{t_i} < \|X\|_s$ for every $X_{t_i}$. Now we explain the definition of $\|X\|_s$. We start by introducing auxiliary notions that are also used in the full proof of Theorem~\ref{thm-main} (see Appendix~\ref{app-main-proof}).

Let $X$ be a set of PCTL state formulae. Recall that $\sub(X)$ denotes the set all state subformulae of all $\varphi \in X$.  The set $\psub(X)$ consists of all path formulae $\Phi$ such that $\sub(X)$ contains a state formula of the form $P(\Phi) \op r$. 

Let $\Phi$ be a path formula of the form $\opf \varphi$ or $\opg \varphi$. We put
\[
    \| \Phi \| ~~=~~ 1 + \sum_{\Psi \in \psub(\varphi)} \| \Psi \|
\]
where the empty sum denotes~$0$. Note that this definition is correct because the nesting depth of $\opf$ and $\opg$ is finite in every path formula.


Now let $X$ be a finite set of PCTL formulae and $s$ a state of a finite Markov chain. Let
\begin{itemize}
  \item $\Deg_s(X)$ be the set of all $\opg \varphi \in \psub(X)$ such that $s \not\models \opg_{=1}\varphi$;
  \item $\cf_s(X)$ be the set of all formulae $\opf \varphi$ such that $X$ contains a formula of the form $\opf_{\op r} \varphi$, $s \not\models \varphi$, and there is a finite path from $s$ to a state~$t$ where $t \models \varphi$ and $t \not\models \opg_{=1}\psi$ for every $\opg \psi \in \Deg_s(X)$ (this is the same definition as in condition~(6) of Definition~\ref{def-loop}, it is recalled for the sake of readability).
\end{itemize}   

\begin{definition}[Progress measure $\| \cdot \|_s$]
\label{def-measure}
Let $X$ be a finite set of PCTL formulae and $s$ a state in a finite Markov chain. We put
\begin{equation*}
        \| X \|_s ~~=~~ 1+ \vert \Deg_s(X) \vert 
                    \cdot \Big( 1 + \sum_{\Phi \in \psub(X)} \| \Phi \| \Big)
                    + \sum_{\Phi \in \cf_s(X)} \| \Phi \|
\end{equation*}
\end{definition}

The progress measure of Definition~\ref{def-measure} appears technical, but it faithfully captures the simplification achieved by a progress loop.

\begin{example}
Let $X = UC_s(\psi)$ for the $\psi$ and $s$ of our running example, i.e.,
  \[
   X \quad = \quad \{\psi,\quad \opg_{=1} \big(\opf_{\geq 0.5}(a \wedge \opf_{\geq 0.2} \neg a) \vee a\big),\quad \opf_{= 1}\, \opg_{=1}\,a ,\quad  \neg a\}
  \]
We have that 
\begin{itemize}
  \item $\Deg_s(X) = \{\opg a\}$,
  \item $\psub(X) = \{ \opg \big(\opf_{\geq 0.5}(a \wedge \opf_{\geq 0.2} \neg a)\vee a\big) ,\ \opf\, \opg_{=1}\,a\}$,
  \item $\cf_s(X) = \emptyset$.
\end{itemize}
Since $\| \opg \big(\opf_{\geq 0.5}(a \wedge \opf_{\geq 0.2} \neg a)\vee a\big)\| = 3$ and 
$\| \opf\, \opg_{=1}\,a \| = 2$, we obtain \mbox{$\| X\|_s = 7$}.
\end{example}

%
%
%


\subsection{Progressive PCTL fragments}
\label{sec-fragments} 

In this section, we give examples of several progressive PCTL fragments. The constraint $\op r$  has the same meaning as in Definition~\ref{def-pctl}, and $\op w$ stands for an arbitrary constraint except for `${=}1$'. 


\begin{description}
\renewcommand{\mid}{~~|~~}
\item[Fragment $\calL_1$]
\[
  \begin{array}{lcl}
     \varphi & ~~::=~~ & a \mid \neg a \mid \varphi_1 \wedge \varphi_2 \mid \varphi_1 \vee \varphi_2 \mid
                     \opf_{\op r} \varphi \mid \opg_{\op r} \psi\\
     \psi    & ::= & a \mid \neg a \mid \psi_1 \wedge \psi_2 \mid \psi_1 \vee \psi_2 \mid
                     \opg_{\op r} \psi
  \end{array}
\]
\item[Fragment $\calL_2$]
\[
  \begin{array}{lcl}
     \varphi & ~~::=~~ & a \mid \neg a \mid \varphi_1 \wedge \varphi_2 \mid \varphi_1 \vee \varphi_2 \mid
                     \opf_{\op r} \varphi \mid \opg_{=1} \psi\\
     \psi    & ::= & a \mid \neg a \mid \psi_1 \wedge \psi_2 \mid \psi_1 \vee \psi_2 \mid
                     \opf_{\op w} \psi
  \end{array}
\]
\item[Fragment $\calL_3$]
\[
  \begin{array}{lcl}
     \varphi & ~~::=~~ & a \mid \neg a \mid \varphi_1 \wedge \varphi_2 \mid \varphi_1 \vee \varphi_2 \mid
                     \opf_{\op r} \varphi \mid \opg_{=1} \psi \mid \opg_{=1} \varrho\\
     \psi    & ::= & a \mid \neg a \mid \psi_1 \wedge \psi_2 \mid \psi_1 \vee \psi_2 \mid
                     \opf_{\op w} \psi\\
     \varrho & ::= & \varrho_1 \wedge \varrho_2 \mid \varrho_1 \vee \varrho_2 \mid
                     \opf_{\op w} \psi \mid \opg_{=1} \psi \mid \opg_{=1} \varrho
  \end{array}
\]
\item[Fragment $\calL_4$]
\[
  \begin{array}{lcl}
     \varphi & ~~::=~~ & a \mid \neg a \mid \varphi_1 \wedge \varphi_2 \mid \varphi_1 \vee \varphi_2 \mid
                     \opf_{\op r} \varphi \mid \opg_{= 1} \psi\\
     \psi    & ::= & a \mid \neg a \mid \psi_1 \wedge \psi_2 \mid \psi_1 \vee \psi_2 \mid
                     \opf_{> 0} \psi \mid \opg_{=1} \psi
  \end{array}
\]
\end{description}
Observe that $\calL_2$ and $\calL_3$ contain the formula $\psi$ of our running example. The above fragments are chosen so that they are not covered by the results of \cite{KR:PCTL-unbounded} and illustrate various properties of Definition~\ref{def-loop}. Fragments $\calL_2$, $\calL_3$, and $\calL_4$ contain formulae requiring non-bottom SCCs with more than one state. 

To demonstrate the applicability of Theorem~\ref{thm-main}, we explictly show that $\calL_2$ is progressive.

\begin{proposition}
  Fragment $\calL_2$ is progressive.
\end{proposition}
\begin{proof}
Let $X \subseteq \calL_2$ be a finite set of formulae and $s$ a state of a finite Markov chain such that $s \models X$ and $X = UC_s(X)$. We show that there exists a progress loop for $X$ and $s$. To achieve that,  
we inductively construct a finite sequence $L_0,\ldots,L_n$, where every $L_i$ is associated to some state $t_i$ reachable from $s$ such that $t_i \models L_i$ . The set $L_0$ is the least set $M$ satisfying the following conditions:
\begin{itemize}
  \item $X \subseteq M$;
  \item if $\varphi_1 \wedge \varphi_2 \in M$, then $\varphi_1,\varphi_2 \in M$;
  \item if $\varphi_1 \vee \varphi_2 \in M$ and $s \models \varphi_1$, then $\varphi_1 \in M$;
  \item if $\varphi_1 \vee \varphi_2 \in M$ and $s \models \varphi_2$, then $\varphi_2 \in M$;
  \item if $\opg_{\op r} \varphi \in M$, then $\varphi \in M$.
  \item if $\opf_{\op r} \varphi \in M$ and $s \models \varphi$, then $\varphi \in M$.
\end{itemize}
We put $t_0 = s$ (observe $s \models L_0$). Furthermore, let $N$ be the set of all formulae $\xi$ such that $\opg_{=1} \xi \in L_0$. 

Suppose that $L_0,\ldots,L_n$ are the sets constructed so far where $t_i \models L_i$ for every $i \in \{0,\ldots,n\}$.  Now we distinguish two possibilities. 
\begin{itemize}
\item If for every formula of the form $\opf_{\op r} \xi \in L_0 \cup \ldots \cup L_n$ where $\opf_{\op r} \varphi \not\in X$ there exists $i \in \{0,\ldots,n\}$ such that $\xi \in L_i$, then the construction terminates.
\item Otherwise, let $\opf_{\op r} \xi \in L_i$ be a formula such that $\opf_{\op r} \xi \not\in X$ and \mbox{$\xi \not\in L_0 \cup \ldots \cup L_n$}. It follows from the definition of the fragment $\calL_2$ that $r \neq 1$. Furthermore, $t_i \not\models \xi$ (this is guaranteed by the closure rules defining $L_0$ and $L_{n+1}$, see below). Since $t_i \models  \opf_{\op r} \xi$, there exists a state $t$ reachable from $t_i$ (and hence also from $s$) such that $t \models \xi$. Furthermore, $t \models N$. Now, we construct $L_{n+1}$, which is the least set $M$ satisfying the following conditions:
\begin{itemize}
  \item $\xi \in M$;
  \item $N \subseteq M$;
  \item if $\varphi_1 \wedge \varphi_2 \in M$, then $\varphi_1,\varphi_2 \in M$;
  \item if $\varphi_1 \vee \varphi_2 \in M$ and $t \models \varphi_1$, then $\varphi_1 \in M$;
  \item if $\varphi_1 \vee \varphi_2 \in M$ and $t \models \varphi_2$, then $\varphi_2 \in M$;
  \item if $\opf_{\op r} \varphi \in M$ and $t \models \varphi$, then $\varphi \in M$.
\end{itemize}
Observe that if $\opg_{=1} \varphi \in M$, then $\opg_{\op r} \varphi \in N$ because $\xi$ does not contain any subformula of the form $\opg_{=1} \varphi$ (see the definition of $\calL_2$). Furthermore, $t \models L_{n+1}$.
\end{itemize}
Note that if  $\opf_{\op r} \varphi \in \DL$, then this formula belongs also to $X$. Now it is easy to see that the constructed $L_0,\ldots,L_n$ is a progress loop. 
\qed
\end{proof}

Let us note that arguments justifying the progressiveness of $\calL_1$ are simple, arguments for $\calL_3$ are obtained by extending the ones for $\calL_2$, and arguments for $\calL_4$ already involve the technical condition~(6) in Definition~\ref{def-loop}.

\section{Conclusions}
\label{sec-concl}

We have shown that the finite satisfiability problem is decidable in doubly exponential space for all PCTL fragments where a progress loop is guaranteed to exist. A natural continuation of our work is to generalize the shape of a progress SCC and the associated progress measure. Attractive candidates are loops with several exit states, and SCCs with arbitrary topology but one exit state. Here, increasing the probability of satisfying $\opf \varphi$ subformulae can be ``traded'' for decreasing the probability of satisfying $\opg \varphi$ formulae, and understanding this phenomenon is another important step towards solving the finite satisfiability problem for the whole PCTL.

Let us note that the technique introduced in this paper can also be used to tackle the decidability of \emph{general} satisfiability for PCTL fragments including formulae that are not finitely satisfiable. By unfolding progress loops into infinite-state Markov chains and arranging the probabilities appropriately, formulae of the form $\opg \varphi$ can be satisfied with arbitrarily large probability by the progress loop 
itself, although the loop is still exited with positive probability. Elaborating this idea is another interesting challenge for future work.

\subsection*{Acknowledgement}
The work is supported by the Czech Science Foundation, Grant No.{}~21-24711S.
%
%
%
%

\bibliographystyle{splncs04}
\bibliography{str-long,concur}

\newpage
\appendix
\section{A proof of Theorem~\ref{thm-main}}
\label{app-main-proof}

Let $X$ be a finite set of PCTL formulae. Let $p(X)$ be the set of path formulae $\Phi$ such that $X$ contains a formula of the form $P(\Phi) \op r$. We use $\nsub(X)$ to denote the subset of $\sub(X)$ consisting of all formulae that are \emph{not} of the form $\opf_{\op r} \xi$ or $\opg_{\op r} \xi$. Furthermore, we put 

\begin{equation*}
   b(X) ~~=~~ 2 + \vert \nsub(X) \vert + \vert \psub(X) \vert + \bigg\vert \bigcup_{\vp \in X} \sub(\vp) \setminus \{\vp\} \bigg\vert.
\end{equation*}  

For a given path formula $\Phi$ and a state $s$, we use $\m(s\models \Phi)$ to denote the probability of all runs initiated in $s$ satisfying $\Phi$. We also use  $\theta_s(X)$ to denote the set of all formulae of the form $P(\Phi) \geq \m(s \models \Phi)$ such that $X$ contains a formula of the form $P(\Phi) \op r$
and $\m(s \models \Phi)>0$.

We start by a sequence of auxiliary observations. The next two lemmas follow directly from definitions.
\begin{lemma}\label{size_norm_decreases}
    Let $\mc$ be a Markov chain, $s,t \in S$ states, and $X$ a
    finite set of formulae such that the following holds:
    \begin{enumerate}
        \item $s \not \models \psi$ for every $\vp \in X$ such
            that $\vp$ is of form $P(\opf \psi) \op r$,
        \item $t \models \psi$ for some $\vp \in X$ such
            that $\vp$ is of form $P(\opf \psi) \op r$,
        \item there is a path from $s$ to $t$.
    \end{enumerate}
    Then $\| UC_t \circ \theta_t (X) \|_t < \| X \|_s$.
\end{lemma}

\begin{lemma}\label{loop_dl_leq_x}
    Let $\fullmc$ be a Markov chain, $s \in S$ a state, and $X$ a
    finite set of formulae, and $\mathcal{L}$
    a progress loop for $M,s,X$. Then $\| \dl \|_s \leq \| X \|_s$.
\end{lemma}

\begin{lemma}\label{size_b_subx}
    Let $\mc$ be a Markov chain, $s \in S$ a state, and $X$ a
    finite set of formulae such that $X = U_s(X)$.
    Then $\vert sub(X) \vert + 1 \leq b(X)$.
\end{lemma}
\begin{proof}
    Clearly, $sub(X) = X \cup \bigcup_{\vp \in X} sub(\vp) \setminus \{\vp\}$.
    Hence, $\vert sub(X) \vert \leq \vert X \vert + \vert \bigcup_{\vp \in X} sub(\vp) \setminus \{\vp\} \vert$.
    For each $\vp \in X$, define $\mu(\vp)$ in the following way:
    \begin{equation*}
        \mu(\vp) = \begin{cases}
                    \Phi & \text{ if } \vp \text{ is of form } P(\Phi) \oplus r, \\
                    \vp  & \text{ otherwise.}
                   \end{cases}
    \end{equation*}
    Observe that $\mu(\vp) \in nsub(X) \cup psub(X)$ for every $\vp \in X$ and $\mu$ can be seen as function
    $\mu \colon X \rightarrow nsub(X) \cup psub(X)$.
    Injectivity of $\mu$ follows from $X = U_s(X)$.
    Hence,  $\vert X \vert \leq \vert nsub(X) \cup psub(X) \vert$.
    It follows that $\vert X \vert \leq \vert nsub(X) \vert + \vert psub(X) \vert$.
    We obtain $\vert sub(X) \vert \leq \vert X \vert + \vert \bigcup_{\vp \in X} sub(\vp) \setminus \{\vp\} \vert$ and also  $\vert X \vert \leq \vert nsub(X) \vert + \vert psub(X) \vert$.
    From this, $\vert sub(X) \vert + 1 \leq b(X)$ follows immediately.
\qed
\end{proof}

A proof of the next lemma follows by a simple calculation.
\begin{lemma}\label{size_b_leq}
    Let $X_1,X_2$ be a finite sets of formulae such that $b(X_1)
    \leq b(X_2)$ and $n_1,n_2 \in \nat$ such that
    $0 < n_1 \leq n_2$. Then
    \begin{equation*}
        2^{b(X_1)}
        \leq 2^{b(X_1)} \cdot \frac{b(X_1)^{n_1}-1}{b(X_1)-1}
        \leq 2^{b(X_2)} \cdot \frac{b(X_2)^{n_2}-1}{b(X_2)-1}.
    \end{equation*}
\end{lemma}

Recall that a run initiated in a BSCC of a finite Markov chain visits all states of this BSCC infinitely often with probability~$1$. Hence, all path formulae are satisfied with probability either $0$ or~$1$. This allows to partition the states of the BSCC according to the satisfaction of path formulae and thus reduce the number of states. This (standard) observation is recalled in the next lemma.
\begin{lemma}\label{prob_bounded_bscc}
    Let $\mc$ be a finite Markov chain, $s \in S$ a state, and
    $X$ a finite set of formulae. Suppose $s \in B$ for some BSCC
    $B \subseteq S$ and $s \models \vp$ for every $\vp \in X$.
    Then there exists a Markov chain $(S',P',v')$ and $s' \in S'$
    such that $\vert S' \vert \leq 2^{|sub(X)|} $ and
    $s' \models \vp$ for every $\vp \in X$.
\end{lemma}






Now we prove the main result.

\begin{theorem}
    Let $\calL$ be a progressive fragment, $\fullmc$ a
    finite Markov chain, $s \in S$ a state, and $X$ 
    a finite subset of $\calL$ such that $X = UC_s(X)$ and $s
    \models \vp$ for every $\vp \in X$. Then there exists a
    Markov chain $M' = (S',P',v')$ and its state $s' \in S'$ such
    that $\vert S' \vert \leq 2^{b(X)} \cdot \frac{b(X)^{\| X
    \|_s+1}-1}{b(X)-1}$ and $s' \models \vp$ for every $\vp \in X$.
\end{theorem}
\begin{proof}

    We use well-founded induction. Formally, we define a relation
    $<$ on the set of pairs of states and finite subsets of
    $\calL$ as follows: $(s_1,X_1) < (s_2,X_2)$ if and only if
    $\| X_1 \|_{s_1} < \| X_2 \|_{s_2}$, where the latter $<$ is
    the standard ordering on $\nat$. It is not hard to see that
    $<$ is well-founded. We now proceed with the proof.

    Since $\calL$ is progressive and $M,s,X$ satisfy the
    required conditions, there exists a progress loop
    $\mathcal{L} = L_0 L_1 \ldots L_n$ for $M,s,X$.

    By the definition of progressive loops, $L_i \subseteq sub(X)$ for all $i \in \{0,1,\ldots,n\}$.
    Then $L_0 \cup L_1 \cup \ldots \cup L_n \subseteq sub(X)$.
    By its definition, $\dl \subseteq L_0 \cup L_1 \cup \ldots \cup L_n$.
    Therefore, $\dl \subseteq sub(X)$.
    Then $\vert \dl \vert \leq \vert sub(X) \vert$.

    Note that $X$ is finite, which means $sub(X)$ is finite, which means $\dl$ is finite.
    Consider $p(\dl)$.
    Since $\dl$ is finite, we know $p(\dl)$ is a finite set of path formulae.
    As sketched out in Section \ref{sec-results},
    there exists a set of states $T \subseteq S$ and $\alpha \colon T \rightarrow [0,1]$ such that the following holds:

    \begin{enumerate}
        \item $\sum_{t \in T} \alpha(t) = 1$,
        \item $0 < \vert T \vert \leq \vert p(\dl) \vert + 1$,
        \item for every $\Phi \in p(\dl)$, it holds that $\m(s \models \Phi) \leq \sum_{t \in T} \alpha(t) \cdot \m(t \models \Phi)$,
        \item for every $t \in T$, there is a path from $s$ to $t$,
        \item for every $t \in T$, it holds that
            $t \in B$ for some BSCC $B \subseteq S$ or
            $t \models \psi$ for some formula $\psi$ such that $\opf \psi \in p(\dl)$.
    \end{enumerate}

    For each $t \in T$, define $X_t = UC_t \circ \theta_t(\dl)$.
    We will show the following holds for every $t \in T$:

    \begin{itemize}
        \item $X_t$ is finite. We have shown $\dl$ is finite, and the rest is easy to see.
        \item $t \models \vp$ for every $\vp \in X_t$. Observe that $t \models \vp$ for every $\vp \in \theta_t(\dl)$ by the definition of $\theta_t$, and the rest is easy to see.
        \item $X_t \subseteq \calL$. Note $X \subseteq \calL$ and $\calL$ is a fragment by the theorem statement. Then $\calL$ is closed under subformulae and changes in probability constraints by our definition of fragments. We have shown $\dl \subseteq sub(X)$. It follows that $\dl \subseteq \calL$, and the rest is easy to see.
        \item $X_t = UC_t (X_t)$. Holds because $UC_t$ is idempotent and $X_t = UC_t \circ \theta_t(\dl)$.
        \item $X_t = U_t (X_t)$. Holds because $U_t$ is idempotent and $X_t = UC_t \circ \theta_t(\dl)$.
        \item $b(X_t) \leq b(X)$. Recall $\dl \subseteq sub(X)$.
            Clearly $sub(\dl) \subseteq sub(X)$.
            Then $nsub(\dl) \subseteq nsub(X)$ and $psub(\dl) \subseteq psub(X)$.
            Similarly, $\bigcup_{\vp \in \dl} sub(\vp) \setminus \{\vp\} \subseteq \bigcup_{\vp \in X} sub(\vp) \setminus \{\vp\}$. 
            It is not too hard to see that application of $\theta_t,C_t,U_t$ preserves these inclusions as well.
            Then, we may conclude that $nsub(X_t) \subseteq nsub(X)$ and $psub(X_t) \subseteq psub(X)$ and also that
            $\bigcup_{\vp \in X_t} sub(\vp) \setminus \{\vp\} \subseteq \bigcup_{\vp \in X} sub(\vp) \setminus \{\vp\}$.
            It follows that $b(X_t) \leq b(X)$.
    \end{itemize}

    We want to show that for every $t \in T$,
    there exists a Markov chain $M_t = (S_t,P_t,v_t)$ and $s_t \in S_t$ such that
    $\vert S_t \vert \leq 2^{b(X)} \cdot \frac{b(X)^{\|X\|_s}-1}{b(X)-1}$ and $s_t \models \vp$ for every $\vp \in X_t$.
    Let $t \in T$.
    By $(5)$, we have $t \in B$ for some BSCC $B \subseteq S$ or $t \models \psi$ for some formula $\psi$ such that $\opf \psi \in p(\dl)$. 

    Suppose $t \in B$ for some BSCC $B \subseteq S$.
    We want to apply Lemma \ref{prob_bounded_bscc}.
    Note that $\fullmc$ is a finite Markov chain, $t \in S$ is a state, and $X_t$ is a finite set of formulae.
    We know $t \in B$ and $B \subseteq S$ is a BSCC, and we have shown $t \models \vp$ for every $\vp \in X_t$.
    Then, by Lemma \ref{prob_bounded_bscc}, there exists a Markov chain $M_t = (S_t,P_t,v_t)$ and $s_t \in S_t$
    such that $\vert S_t \vert \leq \vert \pw(sub(X_t)) \vert$ and $s_t \models \vp$ for every $\vp \in X_t$.
    It follows that $\vert S_t \vert \leq 2^{\vert sub(X_t) \vert}$.
    We want to apply Lemma \ref{size_b_subx}.
    Note that $M$ is a Markov chain, $t \in S$ is a state, and $X_t$ is a finite set of formulae such that $X_t = U_t(X_t)$.
    Then $\vert sub(X_t) \vert + 1 \leq b(X_t)$ by Lemma \ref{size_b_subx}.
    Clearly $\vert sub(X_t) \vert \leq b(X_t)$, which means $2^{\vert sub(X_t) \vert} \leq 2^{b(X_t)}$.
    Then $\vert S_t \vert \leq 2^{b(X_t)}$.
    We want to apply Lemma \ref{size_b_leq}.
    Note that $X_t,X$ are finite sets of formulae such that $b(X_t) \leq b(X)$
    and $\| X \|_s, \| X \|_s$ (we take $n_1=n_2=\|X\|_s$) are natural numbers such that $0 < \| X \|_s \leq \| X \|_s$.
    Then $2^{b(X_t)} \leq 2^{b(X)} \cdot \frac{b(X)^{\|X\|_s}-1}{b(X)-1}$ by Lemma \ref{size_b_leq},
    and so $\vert S_t \vert \leq 2^{b(X)} \cdot \frac{b(X)^{\|X\|_s}-1}{b(X)-1}$.

    Otherwise, $t \models \psi$ for some formula $\psi$ such that $\opf \psi \in p(\dl)$.
    We want to use the induction hypothesis.
    First, note that $\| \dl \|_s \leq \|X\|_s$ by Lemma \ref{loop_dl_leq_x}.
    We want to apply Lemma \ref{size_norm_decreases}.
    We need to show all three conditions of Lemma \ref{size_norm_decreases} hold. 
    We know $\fullmc$ is a Markov chain, $s,t \in S$ are states, and $\dl$ is a finite set of formulae.
    Suppose $\vp \in \dl$ and $\vp$ is of form $P(\opf \psi) \op r$.
    Then $s \not \models \psi$ by the definition of progress loops.
    We need to show existence of $\vp \in \dl$ such that $\vp$ is of form $P(\opf \psi) \op r$ and $t \models \psi$.
    Recall $t \models \psi$ for some formula $\psi$ such that $\opf \psi \in p(\dl)$.
    Since $\opf \psi \in p(\dl)$, we know there exists some $\vp \in \dl$ such that $\vp$ is of form $P(\opf \psi) \op r$.
    Then we have found $\vp \in \dl$ such that $\vp$ is of form $P(\opf \psi) \op r$ and $t \models \psi$. 
    Finally, we know $t \in T$, which means there is a path from $s$ to $t$ by $(4)$. 
    Then, by Lemma \ref{size_norm_decreases}, we have $\| UC_t \circ \theta_t (\dl) \|_t < \|\dl \|_s$.
    In other words, $\| X_t \|_t < \| \dl \|_s$.
    We already know $\| \dl \|_s \leq \| X \|_s$, and so $\| X_t \|_t < \| X \|_s$.
    Clearly $t \in S$ and $X_t$ is a finite subset of $\calL$, and we know $\| X_t \|_t < \| X \|_s$.
    We also know $X_t = UC_t(X_t)$ and $t \models \vp$ for every $\vp \in X_t$.
    Then, by the induction hypothesis, there exists a Markov chain $M_t = (S_t,P_t,v_t)$ and $s_t \in S_t$
    such that $\vert S_t \vert \leq 2^{b(X_t)} \cdot \frac{b(X_t)^{\| X_t \|_t + 1 }-1}{b(X_t)-1}$
    and $s_t \models \vp$ for every $\vp \in X_t$.
    Note that $\| X_t \|_t < \| X \|_s$ implies $\| X_t \|_t + 1 \leq \| X \|_s$.
    Observe that $X_t,X$ are finite sets of formulae such that $b(X_t) \leq b(X)$
    and $\| X_t \|_t + 1, \| X \|_s$ are natural numbers such that $0 < \| X_t \|_t + 1 \leq \|X\|_s$.
    Then $2^{b(X_t)} \cdot \frac{b(X_t)^{\| X_t \|_t + 1 }-1}{b(X_t)-1} \leq 2^{b(X)} \cdot \frac{b(X)^{\|X\|_s}-1}{b(X)-1}$ by Lemma \ref{size_b_leq}.
    Therefore, $\vert S_t \vert \leq 2^{b(X)} \cdot \frac{b(X)^{\|X\|_s}-1}{b(X)-1}$.

    We have established that for every $t \in T$,
    there exists a Markov chain $M_t = (S_t,P_t,v_t)$ and $s_t \in S_t$ such that
    $\vert S_t \vert \leq 2^{b(X)} \cdot \frac{b(X)^{\|X\|_s}-1}{b(X)-1}$ and $s_t \models \vp$ for every $\vp \in X_t$.

    We will now define $M' = (S',P',v')$ and $s' \in S'$ and show they satisfy the properties required by the theorem statement.
    Note that $0 < \vert T \vert \leq \vert p(\dl) \vert + 1$ by $(2)$. 
    We know $p(\dl)$ is finite, and so $T$ is finite and non-empty.
    Then, we can enumerate elements of $T$ so that $T = \{ t_0, t_1, \ldots, t_m \}$.
    For convenience, define $L = \{ L_0, L_1, \ldots, L_n \}$.

    Define $S' = L \cup S_{t_0} \cup S_{t_1} \cup \ldots \cup S_{t_m}$.
    It follows that $\vert S' \vert \leq \vert L \vert + \sum_{t \in T} \vert S_t \vert$.
    Without loss of generality, we may assume that $L,S_{t_0}, S_{t_1}, \ldots, S_{t_m}$ are pairwise disjoint
    (otherwise, we may "rename" states in each $S_{t_i}$ so that they are and alter the corresponding $P_{t_i}$ and $v_{t_i}$ accordingly).

    For every $\kappa \in L$, define $v'(\kappa) = \kappa \cap AP$.
    For every $t \in T$ and $\kappa \in S_t$, define $v'(\kappa) = v_t(\kappa)$.

    For every $t \in T$ and $\kappa \in S_t$ and $\kappa' \in S'$, define $P'(\kappa,\kappa')$ as follows:
    \begin{equation*}
        P'(\kappa,\kappa') =
        \begin{cases}
            P_t(\kappa,\kappa') & \text{ if } \kappa' \in S_t \\
            0                   & \text{ otherwise}
        \end{cases}.
    \end{equation*}
    Intuitively, this means that the transition probabilities for states of $S_{t}$ in $M'$ are the same as in $M_{t}$ for every $t \in T$.
    Note that by the definition of progress loops, $L_0,L_1,\ldots,L_n$ are pairwise distinct.
    For every $i \in \nat$ such that $0 \leq i < n$ and every $\kappa \in S'$, define $P'(L_i,\kappa)$ as follows:
    \begin{equation*}
        P'(L_i,\kappa) =
        \begin{cases}
            1 & \text{ if } \kappa = L_{i+1} \\
            0 & \text{ otherwise}
        \end{cases}.
    \end{equation*}
    Intuitively, this means that every state of the loop transitions to its successor with probability $1$,
    except for the exit state of the loop $L_n$.
    It remains to define transition probabilities for the exit state of the loop, that is, $P'(L_n,\kappa)$ for every $\kappa \in S'$.
    First, we will define $\varepsilon \in \mathbb{R}$, which intuitively represents the probability of remaining in the loop when exiting $L_n$.
    Define $R$ as follows:
    \begin{equation*}
        R = \{ r \in [0,1) \mid \exists \vp \in L_0 \cup L_1 \cup \ldots \cup L_n \text{ such that } \vp \text{ is of form } P(\opf \psi) \op r \}.
    \end{equation*}
    Recall that $L_0 \cup L_1 \cup \ldots \cup L_n \subseteq sub(X)$ and $sub(X)$ is finite.
    It follows that $R$ is finite.
    Now, define $\varepsilon \in \mathbb{R}$ in the following way.
    If $R$ is empty, define $\varepsilon \in \mathbb{R}$ to be an arbitrary real number in $(0,1)$.
    Otherwise $R$ is not empty.
    Then $R$ is a finite and non-empty set of real numbers, and so $R$ has a maximum, which we denote by $max(R)$.
    Since every $r \in R$ is less than $1$, we know $max(R) < 1$.
    Define $\varepsilon \in \mathbb{R}$ to be an arbitrary real number in $(max(R),1)$.
    For every $\kappa \in S'$, define $P'(L_n,\kappa)$ as follows:
    \begin{equation*}
        P'(L_n,\kappa) =
        \begin{cases}
            \varepsilon                 & \text{ if } \kappa = L_0 \\
            (1-\varepsilon) \alpha(t)   & \text{ if } \kappa = s_t \text{ for some } t \in T \\
            0                           & \text{ otherwise}
        \end{cases}.
    \end{equation*}

    We will now show that $\vert S' \vert \leq 2^{b(X)} \cdot \frac{b(X)^{\| X \|_s+1}-1}{b(X)-1}$.
    By its definition, $L = \{ L_0, L_1, \ldots, L_n \}$ and $L_i \subseteq sub(X)$ for every $i \in \{0,1,\ldots,n\}$.
    Then $\vert L \vert \leq \vert \pw(sub(X)) \vert$.
    It follows that $\vert L \vert \leq 2^{\vert sub(X) \vert}$.
    Since $X = UC_s(X)$ by the theorem statement and $U_s$ is idempotent, we have $X = U_s(X)$.
    It is not hard to see that $\vert sub(X) \vert +1 \leq b(X)$ by Lemma \ref{size_b_subx}.
    It follows that $\vert L \vert \leq 2^{b(X)}$.
    We have already shown $\vert \dl \vert \leq \vert sub(X) \vert$.
    Then $\vert \dl \vert +1 \leq \vert sub(X) \vert +1$, which means $\vert  \dl \vert +1 \leq b(X)$.
    Recall that $\vert T \vert \leq \vert p(\dl) \vert + 1$ by $(2)$. 
    Clearly $\vert p(\dl) \vert \leq \vert \dl \vert$.
    It follows that $\vert T \vert \leq b(X)$.

    Recall that $\vert S' \vert \leq \vert L \vert + \sum_{t \in T} \vert S_t \vert$.
    Then $\vert S' \vert \leq 2^{b(X)} + \sum_{t \in T} \vert S_t \vert$ since $\vert L \vert \leq 2^{b(X)}$.
    Since $\vert T \vert \leq b(X)$ and $\vert S_t \vert \leq 2^{b(X)} \cdot \frac{b(X)^{\|X\|_s}-1}{b(X)-1}$ for every $t \in T$,
    it follows that $\vert S' \vert \leq 2^{b(X)} + b(X) \cdot 2^{b(X)} \cdot \frac{b(X)^{\|X\|_s}-1}{b(X)-1}$.
    It is not hard to see that
    \begin{flalign*}
        & 2^{b(X)} + b(X) \cdot 2^{b(X)} \cdot \frac{b(X)^{\|X\|_s}-1}{b(X)-1} \\
        & = 2^{b(X)} + 2^{b(X)} \cdot \frac{b(X)^{\|X\|_s+1}-b(X)}{b(X)-1} \\
        & = 2^{b(X)} \cdot (1 + \frac{b(X)^{\|X\|_s+1}-b(X)}{b(X)-1}) \\
        & = 2^{b(X)} \cdot (\frac{b(X)-1}{b(X)-1} + \frac{b(X)^{\|X\|_s+1}-b(X)}{b(X)-1}) \\
        & = 2^{b(X)} \cdot \frac{b(X)^{\|X\|_s+1}-1}{b(X)-1}. \\
    \end{flalign*}
    In conclusion, $\vert S' \vert \leq 2^{b(X)} \cdot \frac{b(X)^{\|X\|_s+1}-1}{b(X)-1}$.

    It remains to define $s' \in S'$ and show $s' \models \vp$ for every $\vp \in X$.
    Since $\mathcal{L}$ is a progress loop for $M,s,X$, it holds that $X \subseteq L_i$ for some $i \in \{0,1,\ldots,n\}$.
    Define $s'$ to be such $L_i$.
    Clearly $X \subseteq s'$ and $s' \in L$ and $s' \in S'$.
    We want to show $s' \models \vp$ for every $\vp \in X$.
    Since $X \subseteq s'$, it is sufficient to show $s' \models \vp$ for every $\vp \in s'$.
    We will show the following stronger statement:
    $$\forall \vp \in L_0 \cup L_1 \cup \ldots \cup L_n
      \text{ and }
      \forall \kappa \in L
      \text{ it holds that }
      \vp \in \kappa \text{ implies } \kappa \models \vp.
    $$

    We use induction on the structure of $\vp$.
    The cases when $\vp = a$, $\vp = \neg a$, $\vp = \vp_1 \wedge \vp_2$, $\vp = \vp_1 \vee \vp_2$
    follow directly from the definition of progress loops and the induction hypothesis.
    The case when $\vp = P(\opg \psi) \op r$ is analogous to the case when $\vp = P(\opf \psi) \op r$, and
    we show the case when $\vp = P(\opf \psi) \op r$.
    We want to show $\kappa \models P(\opf \psi) \op r$.
    We split the proof into the case when $P(\opf \psi) \op r \not \in \dl$ and the case when $P(\opf \psi) \op r \in \dl$.

    Suppose $P(\opf \psi) \op r \not \in \dl$. Suppose $r<1$.
    It follows that $r \in R$, which means $r \leq max(R)$.
    Since $max(R) < \varepsilon$ by its construction, we have  $r<\varepsilon$.
    Since $P(\opf \psi) \op r \not \in \dl$, we know $\psi \in L_0 \cup L_1 \cup \ldots \cup L_n$.
    Then $\psi \in L_i$ for some $i \in \{0,1,\ldots,n\}$.
    Observe that $\psi \in L_0 \cup L_1 \cup \ldots \cup L_n$ and $L_i \in L$ and $\psi \in L_i$.
    Then $L_i \models \psi$ by the induction hypothesis.
    It is not hard to see that $\m(\kappa \models \opf \psi) \geq \varepsilon$ by construction of $M'$.
    Since $\varepsilon > r$, we have $\m(\kappa \models \opf \psi) > r$.
    It follows that $\m(\kappa \models \opf \psi) \op r$.
    Therefore $\kappa \models P(\opf \psi) \op r$.
    Otherwise $r=1$, which means $P(\opf \psi) \op r$ is $P(\opf \psi) =1$.
    Since $\kappa \in L$, we know $\kappa = L_i$ for some $i \in \{0,1,\ldots,n\}$.
    Since $P(\opf \psi)=1 \not \in \dl$, we know $\psi \in L_i \cup \ldots \cup L_n$.
    Then $\psi \in L_j$ for some $j \in \{i,\ldots,n\}$.
    Observe that $\psi \in L_0 \cup L_1 \cup \ldots \cup L_n$ and $L_j \in L$ and $\psi \in L_j$.
    Then $L_j \models \psi$ by the induction hypothesis.
    It follows that $\m(\kappa \models \opf \psi) = 1$ by construction of $M'$, which means $\kappa \models P(\opf \psi) \op r$.

    Otherwise $P(\opf \psi) \op r \in \dl$.
    Observe that by construction of $M'$, it holds that
    \begin{equation*}
        \m(\kappa \models \opf \psi) \geq \sum_{t \in T} \alpha(t) \cdot \m(M',s_t \models \opf \psi).
    \end{equation*}
    Also by construction of $M'$, it holds that $\m(M',s_t \models \opf \psi) = \m(M_t,s_t \models \opf \psi)$ for every $t \in T$.
    Then
    \begin{equation*}
        \m(\kappa \models \opf \psi) \geq \sum_{t \in T} \alpha(t) \cdot \m(M_t,s_t \models \opf \psi).
    \end{equation*}

    We want to show $\m(M_t,s_t \models \opf \psi) \geq \m(t \models \opf \psi)$ for every $t \in T$.
    Let $t \in T$. If $\m(t \models \opf \psi)=0$, then the statement clearly holds.
    Otherwise $\m(t \models \opf \psi)>0$.
    Recall $P(\opf \psi) \op r \in \dl$, which means $\opf \psi \in p(\dl)$.
    Since $\opf \psi \in p(\dl)$ and $\m(t \models \opf \psi)>0$, we have $P(\opf \psi) \geq \m(t \models \opf \psi) \in \theta_t(\dl)$.
    Recall that $X_t = UC_t \circ \theta_t(\dl)$.
    It is not hard to see that $P(\opf \psi) \geq \m(t \models \opf \psi) \in X_t$.
    Then $M_t,s_t \models P(\opf \psi) \geq \m(t \models \opf \psi)$.
    Therefore, $\m(M_t,s_t \models \opf \psi ) \geq \m(t \models \opf \psi)$.

    We have established that $\m(M_t,s_t \models \opf \psi) \geq \m(t \models \opf \psi)$ for every $t \in T$.
    It follows that
    \begin{equation*}
        \sum_{t \in T} \alpha(t) \cdot \m(M_t,s_t \models \opf \psi) \geq \sum_{t \in T} \alpha(t) \cdot \m(t \models \opf \psi),
    \end{equation*}
    which means that
    \begin{equation*}
        \m(\kappa \models \opf \psi) \geq \sum_{t \in T} \alpha(t) \cdot \m(t \models \opf \psi).
    \end{equation*}
    Recall $P(\opf \psi) \op r \in \dl$.
    Then $\opf \psi \in p(\dl)$.
    Then, by $(3)$, we have
    \begin{equation*}
    \sum_{t \in T} \alpha(t) \cdot \m(t \models \opf \psi) \geq \m(s \models \opf \psi).
    \end{equation*}
    It follows that
    \begin{equation*}
        \m(\kappa \models \opf \psi) \geq \m(s \models \opf \psi).
    \end{equation*}
    Since $P(\opf \psi) \op r \in \dl$, we know $s \models P(\opf \psi) \op r$ by the definition of progress loops.
    That means $\m(s \models \opf \psi) \op r$.
    It follows that $\m(\kappa \models \opf \psi) \op r$, which means $\kappa \models P(\opf \psi) \op r$, which is what we wanted to show.
\end{proof}

\section{Encoding PCTL bounded satisfiability in existential theory of the reals}
\label{app-encoding}

In this section, we sketch a (non-deterministic) polynomial space algorithm deciding bounded PCTL satisfiability. Let $\varphi$ be a PCTL formula and $n \in \N$ a bound on the size of the model. Without restrictions\footnote{Observe that every occurrence of $\opg_{\op r} \varphi$ can be replaced with $\opf_{\triangleleft} \neg \varphi$.}, we assume that $\varphi$ is constructed according to the abstract syntax equation
\[
   \varphi \ ::= \   a \mid \neg a \mid \varphi_1 \wedge \varphi_2 \mid \varphi_1 \vee \varphi_2 \mid
                     \opf_{\bowtie} r
\]
where ${\bowtie} \in \{\geq,>,\leq,<\}$. We disregard the trivial probability constraints `${\geq} 0$',
`${>}1$', `${<}0$', and `$\leq 1$'. 

The algorithm starts by guessing a finite directed graph $(V,{\tran{}})$, where $V = \{v_1,\ldots,v_m\}$ and $m \leq n$. Furthermore, for every subformula $\psi \in \sub(\varphi)$, the algorithm guesses a subset $V(\psi) \subseteq V$ so that 
\begin{itemize}
\item $V(a) = V \setminus V(\neg a)$ for every atomic proposition $a$ such that $\neg a \in \sub(\varphi)$;
\item $V(\xi_1 \wedge \xi_2) = V(\xi_1) \cap V(\xi_2)$ for every $\xi_1 \wedge \xi_2 \in \sub(\varphi)$;
\item $V(\xi_1 \vee \xi_2) = V(\xi_1) \cup V(\xi_2)$ for every $\xi_1 \vee \xi_2 \in \sub(\varphi)$;
\item $V(\varphi) \neq \emptyset$.
\end{itemize}
Then, the algorithm constructs the following formula of existential theory of the reals, where $k = |E|$ and $\fsub(\varphi)$ is the set of all subformulae of $\varphi$ of the form $\opf_{\bowtie r} \psi$.
\[
   \exists x_1,\ldots,x_k ~~:~~ \bigwedge_{i=1}^n 0 < x_i \leq 1 ~~\wedge~~ \bigwedge_{v \in V} \textit{Distr}(v) ~~\wedge~~ \bigwedge_{\psi \in \fsub(\varphi)} \textit{Correct}(V(\psi))
\]

The variables $x_1,\ldots,x_n$ represent the (positive) probability of edges. We write $v_i \tran{x_t} v_j$ to indicate that $x_t$ represents the probability of $v_i \tran{} v_j$.

The formula $\textit{Distr}(v)$ says that the sum of the variables associated with the outgoing edges of $v$ is equal to~$1$, i.e.,
\[
   \sum_{v \tran{x_t} v_j} x_t = 1
\]

The formula $\textit{Correct}(V(\opf_{\bowtie r}\psi))$ says that the set of vertices satisfying the formula $\opf_{\bowtie r}\psi$ is precisely $V(\opf_{\bowtie r}\psi))$, assuming that $V(\psi)$ is correct.

\begin{align*}
   \exists y_1,\ldots,y_n & ~~:~~ \bigwedge_{v_i \in V(\psi)} y_i {=} 1 ~~\wedge~~ 
                                \bigwedge_{v_i \in \nr(V(\psi))} y_i {=} 0\\[1ex]
                          &  ~~\wedge~~
                                \bigwedge_{v_i \in \oth(V(\psi))} y_i = \sum_{v_i \tran{x_t} v_j} x_t \cdot y_j\\[1ex]
                         &       ~~\wedge~~ \bigwedge_{v_i \in V(\opf_{\bowtie r}\psi)} y_i \bowtie r
                                ~~\wedge~~ \bigwedge_{v_i \not\in V(\opf_{\bowtie r}\psi)} y_i \not\bowtie r
\end{align*}     
Here, $\nr(V(\psi))$ is the set of all vertices $v \in V$ such that there is no path from $v$ to a state of $V(\psi)$ in $(V,{\tran{}})$, and $\oth(V(\psi) = V \setminus (V(\psi) \cup \nr(V(\psi)))$. Hence, 
the variable $y_i$ represents the probability of all runs initiated in $v_i$ visiting a vertex in $V(\psi)$.

Observe that the constructed formula belongs to existential theory of the reals and its size is polynomial in the size of $\varphi$ and $n$. Our algorithm outputs `yes' or `no' depending on whether the formula is valid or not (which is decidable in space polynomial in the size of $\varphi$ and $n$ \cite{Canny:Tarski-exist-PSPACE}). Thus, the existence of a model of $\varphi$ with at most $n$ states is decided in polynomial space. 



\end{document}